\documentclass[11pt]{amsart}
\usepackage[english]{babel}

\usepackage{amsmath,amsthm,amsfonts,amssymb, mathrsfs, mathbbol}
\usepackage[foot]{amsaddr}
\usepackage{mathtools}
\usepackage{color}

\usepackage{graphics}
\usepackage{booktabs}

\def\leq{\leqslant}
\def\geq{\geqslant}

\def\text{\mbox}

\def\div{\mathrm{div}\,}

\def\dps{\displaystyle}

\def\R{{\mathbb R}}
\def\Z{{\mathbb Z}}
\def\N{{\mathbb N}}

\def\PP{{\mathbb P}}

\def\P{{\mathbb P}}
\def\PD{\P^\Delta}
\def\S{{\mathbb S}}

\def\pt{\partial}
\def\omm{\Omega}

\def\<{\langle}
\def\>{\rangle}
\def\inclus{{\hookrightarrow}}

\def\grad{\nabla}

\def\hxi{\hat{\xi}}

\def\dim{3}
\def\Rd{\R^\dim}

\def\bve{\;|\;}

\def\fleche{\rightarrow}

\newenvironment{remark}{ {\sc Remark -- } }   {\\}  
\newtheorem{theorem}{Theorem}[section]
\newtheorem{proposition}[theorem]{Proposition}
\newtheorem{lemma}[theorem]{Lemma}

\newcommand{\bfgreek}[1]{\bm{\@nameuse {up#1}}}
\newcommand{\combi}[2]{\dps{ \left(\begin{array}{c} #1\\ #2 \end{array}\right) } }
\def\Harm{{\mathbb H}}

\newcommand{\vc}[1]{ {#1}}
\newcommand{\hypo}[1]{(${\mathscr H}_{#1}$)}

\def\diag{{\rm{diag}}}

\def\Sd{{\mathbb S}^\dim}
\def\Sds{\Sd_\star}

\def\dime{{\rm dim}\, } 
\def\Sdeux{{\mathbb S}^2}

\def\H{H}

\def\dimen{{\rm dim}\,}

\def\YY{{\mathscr Y}}
\def\VV{{\mathscr V}}
\def\ww{{\small {\mathscr W}}}

\def\stm{s^{-1}}

\def\stm{\pi}

\def\Rt{\R^3}

\newcommand{\vecc}[1]{{#1}}

\def\zero{\vecc{0}}
\def\a{\vecc{a}}
\def\e{\vecc{e}}

\def\x{\vecc{x}}
\def\y{\vecc{y}}

\def\omm{\Omega}

\def\curl{{\boldsymbol{\rm curl}\,}}

\def\En{{\mathscr E}}

\newcommand{\pds}[2]{\< #1, #2 \>}

\def\Mg{\vecc{M}}
\def\PMg{\vecc{J}}

\def\Hd{\vecc{H}_d}
\def\vH{\vecc{h}}

\def\Mz{\vecc{M}_0}
\def\n{\vecc{n}}

\def\EEST{{\mathscr E}_{sf}}
\def\EESTH{{\mathscr E}_{sf}}
\def\indom{\chi_\omm}
\def\bomm{\overline{\omm}}
\def\St{{\mathbb S}^2}
\def\mz{\Mg_0}

\def\omm{\Omega}

\def\potf{U}

\def\RAY{{r_0}}

\def\tY{\widetilde{Y}}

\def\INDS{\Lambda}
\def\INDSS{\Lambda^\star}
\author[T. Z. Boulmezaoud]{Tahar Zamene BOULMEZAOUD$^{1, 2}$}
\address{\rm  $^1$ Universit\'e Paris-Saclay, UVSQ, LMV, Versailles, France.}
\address{\rm  $^2$  Department of Mathematics and Statistics, University of Victoria, Victoria, British Columbia, Canada.}
\email{tahar.boulmezaoud@uvsq.fr}

\title[Stray field computation by inverted finite elements]{\Large Stray field computation by inverted finite elements: a new method in micromagnetic simulations}

\title[A simple formula in micormagnetics]{A simple formula of the magnetic potential and of the stray field energy induced by a given magnetization}
\usepackage{shorttoc}


\begin{document}
\begin{abstract}
The primary aim of this paper is the derivation and the proof of a simple  and tractable formula for the stray field energy
in micromagnetic problems. The formula is based on an expansion in terms of Arar-Boulmezaoud functions. 
It  remains valid even if the magnetization is not of constant magnitude or if the sample
 is not geometrically bounded.  The paper continuous with 
a direct and important application which consists in a fast summation technique of the stray field energy. 
The convergence of this technique is established and its efficiency is proved by various numerical experiences.  
 \end{abstract}

\begingroup
\def\uppercasenonmath#1{} 
\let\MakeUppercase\relax 
\maketitle
\endgroup
\setcounter{tocdepth}{1}
\tableofcontents

\section{Introduction}
The description and the understanding of magnetic microstructures are often based 
on the theory of Landau and Lipschitz \cite{LandLif} (see also \cite{brown63}) which consists in minimizing of the total free energy
(see, e. g., \cite{hubert}, \cite{Prohlbook}, \cite{prolhKru2006} and \cite{miyazaki}): 
\begin{equation}\label{totalEnrg}
E_{tot}(\Mg) = \alpha \int_{\omm} | \grad \Mg|^2 dx + \int_{\omm} \phi(\Mg) dx - \frac{1}{2} \int_{\omm} \H_d.\PMg dx - \int_{\omm}  \H_{ex} . \PMg dx + E_s,
\end{equation}
where $\omm$ is the sample (or the magnetic body), $ \alpha $ is the exchange stiffness (positive) constant, $\phi$ is a function describing structural
anisotropies,  $\H_{ex}$ is an external field,  $\Hd$ is  the stray 
(or demagnetizing) field generated by the magnetic body itself  and $E_s$  is the sum the remaining energies (like magnetostrictive self-energy  and magneto-elastic 
 interaction energy). The magnetic polarisation $\PMg$ is given by the formula $\PMg = \mu_0 \Mg$, while  the stray field $\Hd$ is related 
 to $\Mg$ by the equations:
\begin{equation}\label{equaHd}
\curl \Hd = \zero \mbox{ in } \R^3, \; \div (\mu_0 (\Hd + \Mg\chi_{\omm})) = 0 \mbox{ in } \R^3,
 \end{equation}
 where $\chi_{\omm}$ denotes the characteristic function of the sample.    \\
 
 The magnetization
 $\Mg$ is often subject to the Heisenberg-Weiss constraint
\begin{equation}\label{HeiWeiCtr}
|\Mg| = M_s \mbox{ a. e. in } \omm,
\end{equation}
where $M_s$ is the spontaneous saturation magnetization which is assumed to be constant (and generally depending on the temperature). 
Although the reader may assume that $\Mg$ complies with this constraint, we will see that it is not necessary for the validity of the main 
results stated here; a much weaker constraint on $M$ suffices (see assumption \hypo{2} below). \\
 In the litterature, much attention is paid to the calculation of the stray field 
 energy resulting from demagnetizing field $\Hd$: 
\begin{equation}\label{stren1}
\EEST(\Mg) := - \frac{\mu_0}{2} \int_{ \omm} \! \H_d.\Mg dx. 
\end{equation}
In view of equations \eqref{equaHd}, $\Hd$ is curl free and can be written into the form 
\begin{equation}
\Hd = - \grad \potf
\end{equation}
 (see \cite{girault0}),  where $ \potf$ is the magnetic potential which is solution 
of the Poisson equation in {\it the whole space}: 
\begin{equation}\label{main_equa}
\Delta  \potf  =  \div (\Mg \indom)   \mbox{ in } \R^3. 
\end{equation}
The stray field energy can be expressed as
\begin{equation}\label{stren20}
\EEST(\Mg) =   \frac{\mu_0}{2} \int_{\R^3}  \!  |\grad  \potf|^2 dx = \frac{\mu_0}{2} \int_{\R^3} \!   |\Hd|^2 dx. 
\end{equation}
Computing the stray field energy \eqref{stren20} is one of the most 
challenging issues in micromagnetics (see, e. g., \cite{hubert} and \cite{Prohlbook}).  The difficulty is mainly due to its non local nature. There are several methods for the effective calculation of this energy. 
Some of these methods are based on solving the elliptic partial differential equation
 \eqref{main_equa} using finite differences method (see, e. g.,   \cite{Berkov1}, \cite{Vansteenkiste}, \cite{abert1}), or  finite elements method (see, e. g.,  \cite{KoehlerFred1, KoehlerFred2}, \cite{Aurada}, \cite{carstensen2001}), or  inverted finite elements method (\cite{boulmezaoudm2an}, \cite{babatin1},  \cite{boulmezaoud_kerdid_kaliche, KBoul_SF_IFEM}, \cite{boulmezaoud_bhowmik},  \cite{mziou0}, \cite{these_kaliche} and \cite{KBoul_SF_IFEM}). Other methods are based on the calculation of $U$ from the integral formula  (see, e. g.,   \cite{BlueShen},  \cite{longOng},   \cite{Popovi}, \cite{ExlAuz}, \cite{Marty2002}, \cite{labbe}): 
\begin{equation}\label{green_formula}
U(\x) = \frac{1}{4\pi} \int_{\omm}  \!  \frac{(\y-\x).\Mg(\y)}{|\y-\x|^3} d\y. 
\end{equation}
The primary aim of this work is to establish the following  formula
\begin{equation}\label{stren2}
\EEST(\Mg) =   \frac{\mu_0}{2} \sum_{k=0}^\infty \frac{4}{4(k+1)^2-1} \sum_{\alpha \in \INDS_k}  \left(\! \int_{\omm}\Mg.\grad \ww_{\alpha} dx\! \right)^2,
\end{equation}
where $(\ww_\alpha)_{\alpha}$ designate the Arar-Boulmezaoud functions \footnote{Although these functions were discovered by N. Arar and the author in  \cite{ararboulmezaoud1}, the choice of this appellation is not due to the authors, but to a reviewer of one of \cite{arar_boulmez_kk} who asked to choose this appellation.}  introduced in \cite{ararboulmezaoud1} and in \cite{arar_boulmez_kk}. 
These functions will be presented along with their properties in Section  \ref{ABfunc_sec}.  
We also prove the following formula for the magnetic potential  
\begin{equation}\label{u_expression}
 \potf  = \sum_{k=0}^\infty    \sum_{\alpha \in \INDS_k} \frac{4}{4(k+1)^2-1}  \left(\! \int_{\omm}\Mg.\grad \ww_{\alpha} dx\! \right)   \ww_{\alpha}.
\end{equation}
Another by-product, as we shall see, concerns approximation of the stray-field energy \eqref{stren20}. More precisely, 
truncating formula \eqref{stren2} gives the approximation
\begin{equation}\label{stren2N}
\EEST^N(\Mg) =   \frac{\mu_0}{2} \sum_{k=0}^N \frac{4}{4(k+1)^2-1} \sum_{\alpha \in \INDS_k}  \left(\int_{\omm}\Mg.\grad \ww_{\alpha} dx\right)^2,
\end{equation}
When $\Mg . \n = 0$,  we establish the estimate 
\begin{equation}
0 \leq  \EEST(\Mg)-\EEST^N(\Mg)  \leq C N^{-2} \| \div \Mg\|_{L^2(\omm)}. 
\end{equation}
The reader interested in formulas above but not in details of the proof can admit that 
\eqref{stren2} is valid for any connected open set $\omm$, not necessarily bounded, and any 
(measurable) vector field $\Mg$ satisfying
$$
\int_{ \omm} |\Mg|^2 dx < +\infty
$$
Nevertheless, the latter condition is obviously fulfilled when the sample $\omm$ has a finite volume and 
$\Mg$ satisfying the Heisenberg-Weiss constraint \eqref{HeiWeiCtr}. \\
$\;$\\
The rest of the paper is organized as follows. In Section \ref{ABfunc_sec} we present Arar-Boulmezaoud functions
which are the key ingredient of this paper. Their most useful properties are listed.
These properties are essentially known and no originality is claimed in Section \ref{ABfunc_sec}. 
The formulas that form the main output of this paper are presented and proved in Section   \ref{mainRes_sect}.  
In Section \ref{NumMeth}, a new method for calculating the energy resulting from these formulas is proposed
and analyzed. In particular,  the convergence of the method is established. In section \ref{implem_sec} 
focus in on  computational tests through several examples. The last section is devoted to a conclusion. 
\section{Overview of Arar-Boulmezaoud functions}\label{ABfunc_sec}
 In  \cite{ararboulmezaoud1}, Arar and the author  introduced a family of multi-dimensionnal rational and quasi-rational functions
$(\ww_{\alpha})$ which turned out to be particularly appropriate for solving second order elliptic equations
in unbounded regions of space (see \cite{arar_boulmez_kk}). This is primarily due
to their completeness, their orthogonal properties and their behavior at large distances. \\
The definition of these functions in $\R^3$ necessitates the use of spherical harmonics
on the unit sphere of $\R^4$ and the stereographic projection (four dimensional spherical harmonics  are less encoutered than those on $\Sdeux$ the unit sphere of $\R^3$).  \\
For each integer $k \geq 0$, $\Harm_{k}$ will be the space 
of spherical harmonics of degree $k$ over the unit sphere (see, e. g.,
  \cite{seeley}, \cite{muller}, \cite{tomita}, \cite{torres}, \cite{friess}):
$$\Sd := \{ \x \in \R^{4} \mid |\x|=1\}$$
(spherical harmonics of degree $k$ on $\Sd$ are restrictions to $\Sd$  of harmonic 
homogeneous polynomials of degree $k$ on $\R^4$).  We know that
\begin{equation}\label{dim_harm1}
\dime \Harm_{k}= (k+1)^2 \mbox{ for all } k \geq 0.
\end{equation}
In order to construct an orthogonal basis of  $\Harm_{k}$, we set
$$
\INDS = \{(i , \ell, m)  \in \N^2 \times \Z \bve 0 \leq \ell \leq i \mbox{ and } -\ell \leq m \leq \ell\}, 
$$
and for each integer $k \geq 0$
$$
\begin{array}{rcl}
\INDS_k &=& \dps \{(i, \ell, m)  \in \INDS  \bve i = k\}, \\
\INDSS_k&=&  \bigcup_{i=0}^k  \INDS_i. 
\end{array}
$$
If  $\alpha = (k, \ell, m), \; \beta = (k', \ell', m')\in \INDS$, then $\delta_{\alpha, \beta}$ denotes the
usual Kronecker symbol of $\alpha, \beta$, that is $\delta_{\alpha, \beta}= \delta_{k, k'} \delta_{\ell, \ell'} \delta_{m, m'}.$ Define the spherical coordinates for $\S^3$ as the triplet $(\phi, \theta, \chi)$ such that
$0 \leq \phi < 2 \pi$, $0 \leq \theta \leq \pi$, $0 \leq \chi \leq \pi$ and 
\begin{equation}\label{spherical_coordR3}
\xi = (\sin \theta \cos \phi \sin \chi,   \sin \theta \sin \phi \sin \chi,   \cos \theta  \sin \chi,  \cos  \chi),
\end{equation}
Spherical harmonics on $\S^3$ are defined by: 
\begin{equation}\label{defunc_Yalpha}
\YY_{\alpha}  (\xi)= \dps{     \frac{1}{\sqrt{a_{k, \ell}} } (\sin \chi)^\ell T_{k +1}^{(\ell+1)} (\cos \chi) Y_{\ell, m} (\phi, \theta),  } \mbox{  for  }
\alpha = (k, \ell, m) \in \INDS.
\end{equation} 
Here 
\begin{itemize}
\item[--]  $(T_k)_{k \geq 0}$ designate Chebyshev polynomials of the first kind satisfying 
$$
  \cos(k \theta) = T_k(\cos \theta) \mbox{ for } \theta \in \R,
$$
\item[--]  $(Y_{\ell, m})_{\ell, m}$  are the usual real spherical harmonics on $\S^2$: 
\begin{equation}\label{def_YLM} 
Y_{\ell, m} (\phi, \theta) = \eta_{\ell}   K_\ell^{|m|} (\cos \theta)  y_{m}(\phi) 
\end{equation}
where 
\begin{equation}\label{def_YLM_anx} 
y_{m}(\phi) = 
\left\{
\begin{array}{ll}
\dps{   \cos(m \phi)}& \mbox{ if } m \geq 1,  \; \\ 
 \dps \frac{1}{\sqrt{2}} & \mbox{ if } m = 0, \; \\
\dps{  \sin(|m| \phi) }& \mbox{ if } m \leq -1,
\end{array}
\right.  
\end{equation} 
\begin{equation} 
\eta_{\ell} =     \sqrt{\dps{\frac{2 \ell+1}{2\pi }}}, 
\end{equation}
and 
\begin{equation} \label{def_YLM_bis}
K_\ell^m (x) = (-1)^m \sqrt{\frac{(\ell-m)!}{(\ell+m)!}} P_\ell^m (x), \; \mbox{ for }  -\ell \leq m \leq \ell
\end{equation}
(thus, $K_\ell^{-m} = (-1)^m K_\ell^m$).  Here $(P_\ell^m)_{\ell, m}$ designate the associated Legendre functions defined as: 
$$
P_\ell^m (t) = \frac{(-1)^m}{2^\ell \ell!}  (1-t^2)^{m/2} \frac{d^{\ell+m}}{dt^{\ell+m}} (t^2-1)^\ell, \; -\ell \leq m \leq \ell 
$$
(some authors omit the $(-1)^m$ factor, commonly referred to as the Condon-Shortley phase, or append it in the definition of  $Y_{\ell, m}$). 
We also adopt the convention $P_\ell^m = 0$ and $K_\ell^m=0$ when $|m| > \ell$. 
\item[--]  $(a_{k, \ell})$  are  normalization constants given by 
\begin{equation} 
a_{k, \ell} = \frac{(k+1)\pi}{2} \frac{(k+\ell+1)!}{(k-\ell)!}. 
\end{equation}
\end{itemize}
The  following properties hold true
\begin{itemize}
\item For all $k \geq 0$,  $( \YY_{\alpha} )_{\alpha \in \INDS_k }$ is a basis of $\Harm_{k}$.
\item For all $\alpha, \beta \in \INDS$
\begin{equation}
\int_{\S^3} \YY_{\alpha}  (\xi) \YY_{\beta} (\xi)dS(\xi) = \delta_{\alpha, \beta}. 
\end{equation}
\item For all $k \geq 0$ and $\alpha \in  \INDS_k$, 
$$
-\Delta_S \YY_{\alpha} = k(k+2)\YY_{\alpha}, 
$$
where $\Delta_S$ is the Laplace-Beltrami operator over the unit sphere $\S^3$. In terms of  
spherical coordinates, this operator is 
 given  by  
  $$
\dps \frac{1}{\sin^2 \chi}  \left\{ \frac{\pt }{\pt \chi } \left( \sin^2 \chi \frac{\pt }{\pt \chi } \right) +    \frac{1}{\sin \theta}   \frac{\pt }{\pt \theta} \left( \sin \theta \frac{\pt }{\pt  \theta}  \right) +   \frac{1}{\sin^2 \theta}  \frac{\pt^2 }{\pt \phi^2}  \right\}. 
 $$
\end{itemize}
In the three-dimensional situation (the only one that interests us here),   Arar-Boulmezaoud functions are defined as follows (see \cite{ararboulmezaoud1} and  \cite{arar_boulmez_kk}): for any
$\alpha\in \INDS$ 
\begin{equation}\label{express_eigen1}
\ww_{\alpha}(x) =   \left( \frac{2}{|\x|^2+1}\right)^{\frac{1}{2} } \YY_{\alpha} (\stm^{-1}(\x)).
\end{equation}
Here  $\stm$ denotes the 
stereographic projection  defined on $\Sds = \Sd - \{(0,\cdots, 0, 1)$  by
\begin{eqnarray*}
\pi \;: \; \Sds&\longrightarrow&\Rd\\
\vc{\xi}&\longmapsto&  (\frac{\xi_{1}}{1-\xi_{4}},\frac{\xi_{2}}{1-\xi_{4}}, 
\frac{\xi_{3}}{1-\xi_{4}}).
\end{eqnarray*}
Its inverse  is given by
\begin{eqnarray*}
\stm^{-1} \;: \; \Rd&\longrightarrow&\Sds\\
\x&\longmapsto&    \left(\frac{2x_1}{|\x|^{2}+1}, \frac{2x_2}{|\x|^{2}+1}, \frac{2x_3}{|\x|^{2}+1},  \frac{|\x|^{2}-1}{|\x|^{2}+1}  \right).
\end{eqnarray*}
Functions $(\ww_{\alpha})_{\alpha \in \INDS}$ were discovered by Arar and the author in \cite{ararboulmezaoud1} in studying spectrum
of weighted Laplacians in $\R^n$. In Table \ref{table_expli_func} of Appendix B,  the expressions of the first functions $(\ww_{\alpha})_{\alpha \in \INDS}$  are given explicitly. We can then see that these functions have a rational nature. This is a general property as will be announced later.  One can also consult  \cite{ararboulmezaoud1} and  \cite{arar_boulmez_kk}  for higher dimensions and for $n=1$ or $n=2$. \\
In the following Proposition we summarize some useful properties of the functions $(\ww_{\alpha})_{\alpha \in \INDS}$. We refer to \cite{ararboulmezaoud1} and  \cite{arar_boulmez_kk}  for their proofs. 
 \begin{proposition}\label{form_func_basis}
 Let  $k \geq 0$ be an integer and $ \alpha \in \INDS_k$. Then,
 \begin{itemize}
 \item[--] we have
\begin{equation}\label{eigenformula1}
  -\Delta \ww_{\alpha} = \frac{(2k + 1) (2k + 3)}{ (|\x|^2+1)^{2}} \ww_{\alpha},
 \end{equation}
\item[--]   there exists $k+1$ polynomial functions  $p_0, \cdots, p_k$ such that:
\begin{equation}\label{decompo_walpha}
 \ww_{\alpha}  (\x) = \sum_{i=0}^k \frac{p_i(\x)}{(|\x|^2+1)^{i+1/2}},
 \end{equation}
where  for each $ i \leq \ell$, $p_i$ is of degree less than or equal to $i$,
 \item[--]  for all  $\beta \in \INDS$
 \begin{eqnarray}
\dps{ \int_{\Rd} \frac{\ww_{\alpha} (x) \ww_{\beta}(\x) } {(|\x|^2+1)^2} dx }&=& \dps{\frac{1}{4}  \delta_{\alpha, \beta},} \label{ortho_prop_0_1}  \\
\dps{  \int_{\Rd} \grad \ww_{\alpha} (x) . \grad \ww_{\beta}(\x)  dx} &=& \dps{  \frac{ (2k + 1) (2k + 3)}{4} \delta_{\alpha, \beta}.} \label{ortho_prop_0_1_bis}
\end{eqnarray}  
\end{itemize}
\end{proposition}
 The orthogonality identities  \eqref{ortho_prop_0_1}  and \eqref{ortho_prop_0_1_bis} are among the most important properties  of 
Arar-Boulmezaoud functions. as we will see later. More particularly, these relations will be the cornerstone of the formula
 given in this paper and of the resulting numerical approximation. \\
 $\;$\\
Here ends this first enumeration of the properties of functions $(\ww_{\alpha})$. 
We will need other properties later on, in particular for the calculation of 
gradients (see paragraph \ref{calc_grad_sec}).

\section{The first main result: the formulas}\label{mainRes_sect}
The objective here is to prove formulas \eqref{stren2} and \eqref{u_expression} announced in the introduction. 
These formulas will be used 
in the next section to propose a new method for computing stray-field energy. However, before stating the 
first main result, it is appropriate to give some basics concerning
the underlying functional framework we use here. In particular, we show the well-posed nature of 
the equation  \eqref{main_equa}. \\
$\;$\\
Here and subsequently, we assume that 
\begin{itemize}
\item[\hypo{1}] the material fills a connected open set $\omm \subset \R^3$
having a lipschitzian boundary, 
\item[\hypo{2}] the magnetization field $\Mg$ is  defined and  measurable over $\omm$ and satisfies 
\begin{equation}\label{assump_h}
\int_{\omm} |\Mg|^2 dx < \infty, 
\end{equation}
that is $\Mg \in L^2(\omm)^3$.  
\end{itemize}
Assumption \hypo{2} is obviously fulfilled when $|\Mg|$ satisfies the Heisenberg-Weiss constraint \eqref{HeiWeiCtr}
and $|\omm| < \infty$ since
$$
\|\Mg\|^2_{L^2(\omm)^3} = \int_{\omm} |\Mg|^2 dx = |\Mg|^2  |\omm| < +\infty. 
$$
Despite this, we assume neither that $\omm$ is bounded  nor that 
 $|\Mg|$ is satisfying the Heisenberg-Weiss constraint \eqref{HeiWeiCtr}. Only assumptions
  \hypo{1} and \hypo{2} are needed here. \\
We now introduce some weighted function spaces. For all integers $\ell \in \Z$ and $m \geq 0$, 
  $W^m_\ell(\R^3)$ stands for the space of functions $v$ satisfying 
 $$
\forall |\lambda| \leq m, \; (1+|\x|^2)^{(\ell+|\lambda|-m)/2} D^\lambda v \in L^2(\Rt).
$$
This space is equipped with the norm 
\begin{equation}
\|v\|_{W^m_\ell(\Rt)} = ( \sum_{|\lambda| \leq m} \int_{\Rt}  (|\x |^2+1)^{|\lambda|+\ell-m}   |D^\lambda v|^2dx)^{1/2}. 
\end{equation}
When $m \geq 1$, the following inclusions hold: 
$$
W^m_\ell(\R^3) \inclus  W^{m-1}_{\ell-1}(\R^3) \inclus \cdots \inclus W^{1}_{\ell-m+1}(\R^3) \inclus  W^{0}_{\ell-m}(\R^3). 
$$
The following asymptotic property holds true for any function  $v \in W^m_\ell(\R^3)$  (see, e. g., \cite{alliot_these})
\begin{equation}\label{prop_asympt_in_wei}
\lim_{|\x| \fleche + \infty} |\x|^{\ell-m + 3/2} \|v(|\x|, .)\|_{L^2(\St)}   = 0. 
\end{equation}
where $\St$ is the unit sphere of $\R^3$ and 
\begin{equation}\label{defin_normSt}
\|v(|\x|, .)\|^2_{L^2(\St)}  = \int_{\St} |v(|\x|, \sigma )|^2 d\sigma. 
\end{equation}
Let us mention the following Hardy's type inequality in $W^1_0(\R^3)$ (see \cite{KBoul_SF_IFEM}): 
\begin{equation}\label{hardy_inequa_R3}
\forall v \in W^1_0(\R^3), \; \int_{\R^3} \frac{|v|^2}{|\x|^2+1} dx \leq 4 \int_{\R^3} |\grad v|^2 dx. 
\end{equation}
Thus, from now on, we shall consider that the Hilbert space  $W^1_0(\R^3)$ is endowed with the scalar product 
$$
((v, w))_{W^1_0(\R^3)} = \int_{\R^3} \grad v. \grad w dx, 
$$
and with the corresponding norm 
$$
|v|_{W^1_0(\R^3)} = |\grad v|_{L^2(\R^3)},
$$
which is equivalent to the norm $\|.\|_{W^1_0(\R^3)}$.  \\
Here, we look for a solution $\potf$ of \eqref{main_equa} satisfying
\begin{equation}\label{asympto_grad}
 \int_{\R^3} |\grad \potf|^2 dx < \infty, 
\end{equation}
The first main result of this paper is summarized as follows:
\begin{theorem}\label{main_res1}
Assume that assumptions  \hypo{1} and \hypo{2} hold true. Then  \eqref{main_equa} has a unique
solution  $\potf \in W^1_0(\R^3)$ which is given by
\begin{equation}\label{u_expression_th}
\potf  = \sum_{k=0}^\infty    \sum_{\alpha \in \INDS_k} \frac{4}{(2k+1)(2k+3)}  \left(\int_{\omm}\Mg.\grad \ww_{\alpha} dx\right)   \ww_{\alpha}, 
\end{equation}
where the serie in the right-hand side converges in $W^1_0(\R^3)$. The corresponding stray field energy is given by 
\begin{equation}\label{formula_nrg_th}
\EEST(\potf ) =  \sum_{k=0}^\infty \frac{2 \mu_0}{(2k+1)(2k+3)} \sum_{\alpha \in \INDS_k}  \left(\int_{\omm}\Mg.\grad \ww_{\alpha} dx\right)^2. 
\end{equation}
Moreover, 
\begin{enumerate}
\item we have
\begin{eqnarray}
\| (1+|\x|^2)^{-1/2}  \potf \|_{ L^2(\R^3)}  &\leq& 4 \|\Mg\|_{L^2(\omm)}, \label{estima_sol_ord0} \\
\EEST(\potf) &\leq& \dps  \frac{\mu_0}{2} \|\Mg\|^2_{L^2(\omm)}. \label{estima_sol_ord1} 
\end{eqnarray} 
\item  $\potf  \in L^2(\R^3)$, $(1+|\x|^2)^{1/2} \grad \potf  \in L^2(\R^3)^3$ and 
\begin{equation}\label{estima_sol_ord11}
\|\potf\|_{L^2(\R^3)}  +    \|(1+|\x|^2)^{1/2} \grad \potf\|_{L^2(\R^3)^3} \leq C_0(\omm) \|\Mg\|_{L^2(\omm)},  
\end{equation}
for some constant $C_0(\omm) > 0$ depending only on $\omm$. 
\item we have 
\begin{equation}\label{lim_infini_U}
\lim_{|\x| \fleche + \infty} |\x|^{3/2}  \|\potf(|\x|,  .)\|_{L^2(\St)} = 0,
\end{equation}
\end{enumerate}
\end{theorem}
Issues concerning 
the regularity of the solution $U$ are postponed to next section (see Theorem  \ref{conver_theorem}). 
\begin{proof}
We can reformulate equation \eqref{main_equa}  as follows: find $\potf \in W^1_0(\R^3)$ such that
\begin{equation}\label{forma_weak}
\forall v \in W^1_0(\R^3), \; \int_{\R^3} \grad \potf . \grad v dx = \int_{\omm} \Mg.\grad v dx. 
\end{equation}
The existence and uniqueness of solutions is a direct consequence of the Lax-Milgram theorem. 
Estimate \eqref{estima_sol_ord1} results from the use of Cauchy-Schwarz inequality on the right when $v=u$ in \eqref{forma_weak}. 
Combining with Hardy inequality \eqref{hardy_inequa_R3} gives \eqref{estima_sol_ord0}.  \\
We also have the following lemma (see \cite{ararboulmezaoud1} and \cite{arar_boulmez_kk}): 
\begin{lemma}\label{hilb_basis}
The family $(\ww_\alpha)_{\alpha \in \INDS}$ is a Hilbert basis of $W^1_0(\R^2)$ endowed with the norm $|.|_{W^1_0(\R^3)}$. 
\end{lemma}
Thus,
$$
\begin{array}{rcl}
\potf & =& \dps \sum_{\alpha \in \INDS}   \frac{  ((\potf, \ww_{\alpha} ))_{W^1_0(\R^3)}  }{  (( \ww_{\alpha},  \ww_{\alpha} ))_{W^1_0(\R^3)}   }    \ww_{\alpha}, \\
& =& \dps \sum_{k=0}^{+\infty}   \frac{ 4   }{   (2k+1)(2k+3) } \sum_{\alpha \in \INDS_k} (\grad \potf,  \grad \ww_{\alpha} )_{L^2(\R^3)^3}  \ww_{\alpha}.
\end{array}
$$
Combining with \eqref{forma_weak}  and  \eqref{ortho_prop_0_1_bis} gives \eqref{u_expression_th}. Since convergence
of the right-hand side holds in $W^1_0(\R^3)$ we also get \eqref{formula_nrg_th}.  The reader can refer to \cite{KBoul_SF_IFEM} for estimate \eqref{estima_sol_ord11}. Hence, $\potf \in W^1_1(\R^3)$ and  \eqref{prop_asympt_in_wei} holds true with $\ell= m=1$. This gives
\eqref{lim_infini_U}. 
\end{proof}

\section{The second main result: a new method for calculating the stray-field energy}\label{NumMeth}
The main purpose of this section is to show that from the two formulas  \eqref{u_expression_th} and \eqref{formula_nrg_th} 
 results  a very efficient and easy to implement numerical method for calculating the stray field energy.
This numerical method could be seen as a spectral method in an unbounded domain. However,  unlike the usual spectral methods in a bounded domain and which use polynomial functions or trigonometric functions, here we use (quasi)-rational functions
   guaranteeing a decay of the solution at large distances. Indeed, in view of  Proposition \ref{form_func_basis}, the functions $(\ww_\alpha)_\alpha$ are rationals up to a multiplicative factor. 
\subsection{The method}\label{sect_meth}
In view of \eqref{formula_nrg_th}, the energy $\EEST(\potf) $ can be reasonably approximated by truncating the sum. For this end,
we set for each $N \geq 1$
\begin{equation}\label{discrete_N_energy}
\EEST^N(\potf) =   \sum_{k=0}^N\sum_{\alpha \in \INDS_k}  \frac{2 \mu_0}{(2k+1)(2k+3)} \left(\int_{\omm}\Mg.\grad \ww_{\alpha} dx\right)^2. 
\end{equation}
We observe that 
\begin{equation}
\EEST^N(\potf) =\frac{\mu_0}{2} \int_{\R^3} |\grad \potf_N|^2 dx, 
\end{equation}
where 
\begin{equation}\label{uN_expression}
\potf_N  = \sum_{\ell=0}^N    \sum_{\alpha \in \INDS_k} \frac{4}{(2k+1)(2k+3)}  \left(\int_{\omm}\Mg.\grad \ww_{\alpha} dx\right)   \ww_{\alpha}.
\end{equation}
Let us give another interpretation of $\potf_N$. Define the family of finite dimensional spaces  $(H_N)_{N \geq 0}$ as follows: for $N \geq 0$, $H_N$ 
is the space of functions of the form
\begin{equation}\label{pseudoform2}
v(\x) = \sum_{k=0}^N \frac{p_k(\x)}{(|\x|^2+1)^{k+1/2}}, \; \x \in \Rd, 
\end{equation}
where, for each $k \leq N$, $p_k$ is a polynomial  of degree less than or equal to $k$.   Obviously, 
\begin{equation}
H_0 \subset H_1 \subset H_2 \subset \cdots \subset H_N \subset \cdots
\end{equation}
The following inclusion holds for $N \geq 0$: 
\begin{equation}
H_N \inclus W^1_0(\R^3). 
\end{equation}
It can be easily proved that (see, e. g., \cite{ararboulmezaoud1})
\begin{equation}
\dimen  H_N =
\combi{3+N}{3} + \combi{N+2}{3}  =   \frac{(N+1)(N+2)(2N+3)}{6}. 
\end{equation}
Since 
$$
|\INDS_k| = (k + 1)^2 \mbox{ for } k \geq 0
$$
($|\INDS_k|$ designates the cardinal of the set $\INDS_k$), 
we deduce the identity 
\begin{equation}\label{dimeHNequa}
\dimen  H_N = \sum_{k=0}^N |\INDS_k| = |\INDSS_N|. 
\end{equation}
On the other hand, in view of Proposition \ref{form_func_basis}, we have for all $N \geq 0$,
$$
(\alpha \in \INDS_\ell \mbox{ for some } \ell \leq N) \Longrightarrow \ww_{\alpha} \in H_N. 
$$
In other words, 
$$
 \{ \ww_{\alpha}; \; \alpha \in  \INDSS_N \} \subset H_N. 
$$
Combining the latter  with \eqref{dimeHNequa} and with orthogonality properties \eqref{ortho_prop_0_1} and  \eqref{ortho_prop_0_1_bis} gives 
\begin{lemma}
For all $N \geq 1$, the family $\dps{ (\ww_{\alpha})_{\alpha \in \INDSS_N}}$ is  a basis of $H_N$. 
\end{lemma}

Now, we state this
\begin{proposition}
The function $\potf_N$ given by formula \eqref{uN_expression} is also the unique solution of the well-posed discrete problem
\begin{equation}\label{discretepb1}
\forall v_N \in H_N,  \; \; \int_{\R^3} \grad \potf_N .\grad v_N dx  =\int_{\omm} \Mg. \grad v_N dx. 
\end{equation}
In addition, $\potf_N$ is the projection of $\potf$ on $H_N$ with respect to the inner product $((.,.))_{W^1_0(\R^3)}$. 
\end{proposition}
One could therefore consider that the approximation  \eqref{uN_expression} is none other than the solution of 
the discrete problem \eqref{discretepb1}  which consists to approximate the original problem \eqref{main_equa} 
by a spectral method  using the functions of $H_N$.  The use of the family 
$(\ww_\alpha)_{\alpha \in \INDSS_N}$  as a  basis of $H_N$  reduces  the discrete problem \eqref{discretepb1} 
to a simple {\it diagonal linear system} 
\begin{equation}
 D X = B
 \end{equation}
 with $D$ the diagonal matrix
$$
D = \diag(\frac{4}{3}, \underbrace{\frac{4}{15}, \cdots,  \frac{4}{15}}_{ \text{4 coefficients} }, \cdots,  \underbrace{\frac{4}{4(N+1)^2-1}, \cdots,  \frac{4}{4(N+1)^2-1}}_{\text{$(N+1)^2$ coefficients}}).
$$
Here $X$ contains the components of  $\potf_N$ with respect to the basis   $(\ww_\alpha)_{\alpha \in \INDSS_N}$
and $B$ covers the integrals $\dps{\int_{\omm} \Mg. \grad \ww_\alpha dx.}$ Thus, solution of \eqref{discretepb1}
is obviously given by formula \eqref{u_expression_th}. This is a significant observation which  demonstrates 
the benefits of using functions $(\ww_\alpha)_{\alpha}$. \\

\subsection{Convergence of the method and error estimate} $\;$\\
Focus now is on convergence when $N \to +\infty$. We have: 
\begin{lemma}\label{third_main_th}
Assume that \hypo{1} and \hypo{2} are fullfilled. Then,
\begin{equation}\label{energy_gradnrm}
 \EEST(\potf) - \EEST(\potf_N) = \frac{\mu_0}{2}  |\potf  -\potf_N|^2_{W^1_0(\R^3)}. 
\end{equation}
and 
\begin{equation}\label{convergence_lim}
\lim_{N \to +\infty} \EEST^N(\potf_N)  = \EEST(\potf).  
\end{equation}
\end{lemma}
\begin{proof}
We first observe that \eqref{convergence_lim} is a direct consequence 
of \eqref{u_expression_th}. Indeed, in view of \eqref{forma_weak} and  \eqref{discretepb1} we get 
$$
\int_{\R^3} (\grad \potf -  \grad \potf_N). \grad \potf_N dx=0,
$$
and \eqref{energy_gradnrm}  follows immediately. 
\end{proof}

\begin{theorem}\label{conver_theorem}
Assume that \hypo{1} and \hypo{2} are fullfilled. Assume also  that $\omm$ is bounded, $\div \Mg \in L^2(\omm)$ and  $\Mg . \n = 0$ on $\pt \omm$. Then, $\potf \in W^{2}_{2}(\R^3)$ and there exists a
constant $C_1$ depending only on $\omm$ such that 
\begin{eqnarray}
\|\potf - \potf_N\|_{W^1_0(\R^3)} &\leq&  \frac{C_1}{N}\|\div \Mg\|_{ L^2(\omm)},   \label{estima_norm_k1}\\
0 \leq \EEST(\potf) - \EEST(\potf_N) &\leq&  \frac{C^2_1}{N^2} \|\div \Mg\|^2_{ L^{2}(\omm)}. \label{estima_nrj_conv1}
\end{eqnarray}
If in addition $\div \Mg \in H_0^{k-1}(\omm)$ for 
some integer $k \geq 2$ and  if 
\begin{equation}\label{cond_exist_M}
\int_{\omm} \Mg . \grad q dx = 0 \mbox{ for all } q \in \PD_{k-1}, 
\end{equation}
then  $\potf \in W^{k+1}_{2k}(\R^3)$ and there exists a constant
$C_k$ depending only on $k$ and $\omm$ such that 
\begin{eqnarray} 
\|\potf - \potf_N\|_{W^1_0(\R^3)} &\leq&  \dps{ C_k N^{-k} \|\div \Mg\|^2_{{H}^{k-1}(\omm)}}, \label{estima_norm_conv} \\
0 \leq \EEST(\potf) - \EEST(\potf_N) &\leq& C^2_k N^{-2k} \|\div \Mg\|^2_{ {H}^{k-1}(\omm)} \label{estima_nrj_conv2}
\end{eqnarray}
\end{theorem}
Here, the usuel Sobolev space  $H_0^{k-1}(\omm)$ designates  the closure of ${\mathscr C}^\infty_0(\omm)$ in 
the usual Sobolev space $H^{k-1}(\omm)$. 
\begin{proof}
Firstly, we adopt the following notation: given a function $f$ defined over $\omm$, we denote
by $\widetilde{f}$ its extension to $\R^3$ defined as 
$$
\widetilde{f} =
\left\{
\begin{array}{ll}
f & {\rm in }  \; \; \omm,\\
0 & {\rm in } \; \;  \R^3 \backslash \omm. 
\end{array}
\right. 
$$
The following lemma is due to \cite{AmroucheGiroireGirault} (Theorem 6.6): 
\begin{lemma}\label{giroireLap}
Let $m \geq 1$ and $\ell \geq 1$ be two integers. Then, the Laplace operator $\Delta$ defined by
$$
\Delta: \; W^{1+m}_{\ell+m} (\R^3) \to W^{-1+m}_{\ell+m} (\R^3) \perp \PD_{\ell-1}, 
$$
is an isomorphism. Here  $ \PD_{\ell-1} = \{ p \in  \PP_{\ell-1}  \bve \Delta p = 0 \}$ and 
$$
W^{-1+m}_{\ell+m} (\R^3) \perp \PD_{\ell-1} = \{ f \in W^{-1+m}_{\ell+m} (\R^3) \bve \int_{\R^3}  f q dx = 0  \mbox{ for all }  q \in \PD_{\ell-1}  \}. 
$$
\end{lemma}  
Assume now that $\div \Mg \in L^2(\omm)$ and $\Mg . \n = 0$ on $\pt \omm$.  Then, $\potf$ is solution of the problem
\begin{equation}
\Delta \potf = \widetilde{\div \Mg} \mbox{ in } \R^3. 
\end{equation}
Obviously $ \widetilde{\div \Mg} \in W^0_2(\R^3)$ and 
$$
\int_{\R^3} \widetilde{\div \Mg} dx = \int_{\omm} \div \Mg dx = 0. 
$$
In view of condition \eqref{cond_exist_M} and Lemma \ref{giroireLap}, we deduce that $\potf \in W^{2}_{2}(\R^3)$. 
If in addition   $\div \Mg \in {H}_0^{k-1}(\omm)$
for some $k \geq 1$ and if $\Mg$ satisfies condition \eqref{cond_exist_M} when $k \geq 2$, then
  $\widetilde{\div \Mg} \in W_s^{k-1}(\R^3)$ for any real number $s$ (since 
$\widetilde{\div \Mg}$ vanishes outside $\omm$). In particular $\widetilde{\div \Mg} \in W_{2k}^{k-1}(\R^3)$. By Green's formula
we also have
$$
\forall q \in \PD_{k-1}  \int_{\R^3}\widetilde{\div \Mg} q dx = \int_{\omm}{\div \Mg} q dx  = -  \int_{\omm}{\Mg}.\grad q dx = 0. 
$$
Hence, $\potf \in W^{k+1}_{2k}(\R^3)$, thanks to Lemma \ref{giroireLap}. Moreover, there exists
a constant $C_k$ depending only on $k$ such that
\begin{equation}\label{estima_u_div}
\| \potf \|_{W^{k+1}_{2k}(\R^3)} \leq C_k \| \widetilde{ \div \Mg}\|_{W^{k-1}_{2k}(\R^3)} \leq \widetilde{C}_k  \| \div \Mg\|_{H^{k-1}(\omm)}. 
\end{equation}
Let $\pi_N$ be the orthogonal projector on $H_N$ with respect to the scalar product
associated to the norm $|.|_{W^1_0(\R^3)}$. 
The following result is due to \cite{arar_boulmez_kk}: 
\begin{lemma}
Assume that $v \in W^{k+1}_{2k}(\R^3)$ for some integer $k \geq 0$. Then,
\begin{equation}\label{estima_erro_vpnv}
\|\grad v - \grad (\pi_N v)\|_{L^2(\R^3)^3} \leq C_k^\star  N^{-k} \|v \|_{W^{k+1}_{2k}(\R^3)},
\end{equation}
where $ C_k^\star$ is a constant which depends neither on $N$ nor on $v$. 
\end{lemma}
We know that $\potf_N = \pi_N U$. The  inequalities  \eqref{estima_norm_k1} and \eqref{estima_nrj_conv1} 
result from  \eqref{estima_erro_vpnv} and \eqref{estima_u_div} with $k=1$ and from  \eqref{hardy_inequa_R3}. 
 The  inequalities  \eqref{estima_norm_conv}  and \eqref{estima_nrj_conv2} are deduced in a similar way. 

%
%
\end{proof}
\begin{remark}\label{regularity_rem} 
Assumption $\Mg . \n = 0$ on $\pt \omm$ means that the effective magnetic charges are zero. One can easily see
 that if  $\Mg . \n \neq 0$ on $\pt \omm$   then $\potf$  does not belong to $W^{2}_{2}(\R^3)$. Indeed, equation
 \eqref{main_equa}  can be rewritten as 
$$
\left\{
\begin{array}{rcll}
\Delta u &= & \div \Mg &\mbox{ in } \omm,  \\
\Delta u &= & 0 & \mbox{ in }\R^3 \backslash \overline{\omm},  \\
\dps{[u]} &=& 0 &  \mbox{ on } \pt \omm, \\
\dps{ \left[ \frac{\pt u}{\pt n} \right] }&=& - \Mg.\n & \mbox{ on } \pt \omm, 
\end{array}
\right.
$$
where $\n$ is the exterior normal on $\pt \omm$.  Thus, $\dps  \left[ \frac{\pt u}{\pt n} \right] \ne 0$  on  $\pt \omm$ and
$\potf \not \in W^{2}_{2}(\R^3)$. 
\end{remark}

\section{Implementation and computational tests}\label{implem_sec}
The first purpose of this section is to examine the numerical results obtained after implementation of the 
method suggested in the previous section and to check whether the theoretical error estimates 
are confirmed numerically and whether they are optimal. Another goal is to give some additional details 
regarding the implementation of the method, including the calculation of integrals. 
It is worth noting at this early stage that despite the three-dimensional nature of the problem, and despite the fact that it is posed in an 
open domain, the implementation of the method remains rather easy and fast.

\subsection{Additional details about gradients of Arar-Boulmezaoud functions}\label{calc_grad_sec} $\;$\\
Formulas in Theorem \ref{main_res1} as well as the approximation method proposed in Section \ref{sect_meth} involve functions  $\dps{ (\ww_{\alpha})_{\alpha}}$  by their gradients, particularly in the integral coefficients
\begin{equation}\label{coeff}
\int_{\omm}\Mg.\grad \ww_{\alpha} dx.
\end{equation}
In practice, during the implementation of the method, 
the precise calculation of these gradients could be of great importance.
It is consequently preferable to compute 
 them by exact analytical expressions and not by discretization of the differentiation operators. 
Of course, one can use a Green's formula in \eqref{coeff} to make these gradients disappear:
$$
\int_{\omm}\Mg.\grad \ww_{\alpha} dx = -\int_{\omm}(\div \Mg)  \ww_{\alpha} dx  + \pds{\Mg.\n}{ \ww_{\alpha}}_{\pt \omm},
$$
($ \pds{.}{.}_{\pt \omm}$ designates the duality pairing between $H^{1/2}(\pt \omm)$ and $H^{-1/2}(\pt \omm)$).  
 However, this requires  a little more regularity on the magnetization vector field $\Mg$ (for example that $\div \Mg \in L^2(\omm)$) and, moreover,  
it makes surface integrals appear. 
It will therefore not be useless to spell out the gradients 
 $\dps{ (\grad \ww_{\alpha})_{\alpha}}$. Actually, in view of \eqref{defunc_Yalpha} and \eqref{express_eigen1}, 
 these gradients are not quite easy to calculate, 
 especially because of the special functions appear in their formulas (that is, Chebyshev polynomials 
 and  associated Legendre functions of Legendre).  \\
 In this paragraph, we deduce simpler and exact expressions to the gradients
 of the functions  $\dps{ (\ww_{\alpha})_{\alpha}}$, in order to facilitate
 the computation of magnetic potential and the stray-field energy by formulas
   \eqref{uN_expression} and \eqref{discrete_N_energy}. \\
$\;$\\
The starting point is the following proposition
\begin{proposition}\label{formula_of_gradient_lem_bis}
For  $\alpha \in \INDS$ and $\x \in \R^3$:
\begin{equation}\label{formula_of_gradient_ww_bis}
\grad \ww_\alpha (\x)  =( 1-\xi_4)^{1/2} ( \VV_\alpha (\xi) - \frac{1}{2} \YY_\alpha(\xi) \hxi), 
\end{equation}
where  $\xi =  (\xi_1, \xi_2, \xi_3, \xi_4) = \pi^{-1}(\x) \in  \S^3$, $\hxi = (\xi_1, \xi_2, \xi_3)$,  $(\phi, \theta, \chi)$ are the spherical coordinates of $\xi$  (see  \eqref{spherical_coordR3}) and 
\begin{equation}\label{form_VValpha}
\VV_\alpha(\xi)  =  (1-\cos \chi)
\left(\!\! 
\begin{array}{cccc}
 -\dps{ \sin  \phi  }    & \dps{ \cos \phi \cos \theta    } &  -\cos \phi \sin \theta \\
  \dps{ \cos \phi  }  & \dps{ \sin \phi \cos \theta  } &   -\sin \phi \sin \theta \\
 0 &-\dps{ \sin \theta  } & - \cos \theta\\
\end{array}
\!\! \right) 
\left(\!\! 
\begin{array}{c}
 \dps{ \frac{1}{\sin \theta \sin\chi}  \frac{\pt  \YY_\alpha}{ \pt  \phi}  (\xi)}  \\
 \dps{  \frac{1}{ \sin\chi}    \frac{\pt  \YY_\alpha}{ \pt  \theta} (\xi) }\\
 \dps{  \frac{\pt  \YY_\alpha}{ \pt  \chi} (\xi)} 
\end{array}
\!\! \right).
\end{equation}
\end{proposition}
By the sake of simplicity, proof of Proposition \ref{formula_of_gradient_lem_bis} is postponed to Appendix A.  \\
$\;$\\
\begin{remark}
In Proposition \ref{formula_of_gradient_lem_bis}, $\frac{\pt  \YY_\alpha}{ \pt  \phi} $, $\frac{\pt  \YY_\alpha}{ \pt  \theta}$
and $\frac{\pt  \YY_\alpha}{ \pt  \chi}$  designate (abusively) the derivatives of $\YY_\alpha$ considered as a function
of $\theta$, $\phi$ and $\chi$. 
\end{remark}
$\;$\\
At this stage, all that remains is the calculation of the partial derivatives 
$$
 \frac{\pt  \YY_\alpha}{ \pt  \phi}, \;   \frac{\pt  \YY_\alpha}{\pt  \theta} \mbox{ and }  \frac{\pt  \YY_\alpha}{ \pt  \chi}. 
$$
In view of formula  \eqref{defunc_Yalpha}, the first two ones can be easily 
expressed in terms of derivatives of spherical harmonics on $\S^2$. For example,
 if $\alpha = (k, \ell, m)$ then
\begin{eqnarray}
\dps{ \frac{\pt \YY_{\alpha}}{\pt \phi} (\xi) }&= &\dps{     \frac{1}{\sqrt{a_{k, \ell}} } (\sin \chi)^\ell T_{k +1}^{(\ell+1)} (\cos \chi) \frac{\pt Y_{\ell, m}}{\pt \phi}  (\phi, \theta) } \nonumber\\
 &= &\dps{   -    \frac{m}{\sqrt{a_{k, \ell}} } (\sin \chi)^\ell T_{k +1}^{(\ell+1)} (\cos \chi)  Y_{\ell, -m}  (\phi, \theta)} 
\end{eqnarray}
In order to avoid division by zero in \eqref{form_VValpha} (when $\sin \theta = 0$), 
which is useless,  one can employ in the definition \eqref{def_YLM}  of $Y_{\ell, -m}$  the  
recurrence property on associated Legendre functions: 
\begin{equation}\label{formu_sansdiv}
2 m  K_{\ell}^m(\cos \theta) =  \sin \theta  \left (\tau_{\ell, m}  K_{\ell+1}^{m+1}(\cos \theta) + \tau_{\ell, -m} K_{\ell+1}^{m-1}(\cos \theta)\right), 
\end{equation}
with 
$$
\tau_{\ell, m} = \sqrt{(\ell+m+2)(\ell+m+1)} \mbox{ for }-\ell \leq m \leq \ell. 
$$
Thus, for $m \ne 0$ we have 
\begin{equation}
\begin{array}{rcl}
\dps{ \frac{2 |m| }{\sin \theta} \frac{\pt \YY_{\alpha}}{\pt \phi} (\xi) } &=&  - \dps{ m    \frac{  \eta_{\ell} }{\sqrt{a_{k, \ell}} }     (\sin \chi)^\ell T_{k +1}^{(\ell+1)} (\cos \chi)  y_{-m}(\phi)  } \\
&& \biggl( \tau_{\ell, |m|}  K_{\ell+1}^{|m|+1}(\cos \theta) + \tau_{\ell, -|m|} K_{\ell+1}^{|m|-1}(\cos \theta)\biggr) \\
\end{array}
\end{equation}
Similarly, we have 
\begin{equation}
\begin{array}{rcl}
 \dps \frac{\pt \YY_{\alpha}}{\pt \phi}(\xi)& = & \dps{  \frac{ \eta_{\ell}}{\sqrt{a_{k, \ell}} } (\sin \chi)^\ell T_{k +1}^{(\ell+1)} (\cos \chi)  }\\
 &&  \dps \biggl(  c_{\ell, -|m|}   K_\ell^{|m|-1}(\cos \theta) - c_{\ell, |m|}  K_\ell^{|m|+1} (\cos \theta) \biggr)   y_m(\phi), 
\end{array}
\end{equation}
where  
\begin{equation}
 c_{\ell, m} = \frac{1}{2} \sqrt{(\ell -m) (\ell+m+1)} \mbox{ for  } \ell \geq 0 \mbox{ and  } -\ell \leq m \leq \ell. 
\end{equation}
Note that we used the following recurrence formula:
\begin{equation}\label{formula1KLM}
(\sin \theta) (K_\ell^m)'(\cos \theta) =  c_{\ell, m}  K_\ell^{m+1} (\cos \theta) -  c_{\ell, -m}   K_\ell^{m-1}(\cos \theta). 
\end{equation}
(and with the convention $K_\ell^{j} = 0$ when $|j| > \ell$). Hence
\begin{equation}
\begin{array}{rcl}
 \dps \frac{\pt \YY_{\alpha}}{\pt \theta}(\xi) &=&  \dps{  \frac{ \eta_{\ell}}{\sqrt{a_{k, \ell}} } (\sin \chi)^\ell T_{k +1}^{(\ell+1)} (\cos \chi)  }\\
&&\dps{   \biggl(  c_{\ell, -|m|}   K_\ell^{|m|-1}(\cos \theta) - c_{\ell, |m|}  K_\ell^{|m|+1} (\cos \theta) \biggr)   y_m(\phi). }
\end{array}
 \end{equation}
 Finally, we also have 
\begin{equation}\label{first_formula_drv_chi3}
\begin{array}{rcl}
\dps{ \frac{\pt \YY_{\alpha}}{\pt \chi} (\xi)} &=& \dps \frac{  (\sin \chi)^{\ell-1}}{\sqrt{a_{k, \ell}} }  Y_{\ell, m} (\phi, \theta), \\
&&   \left( \ell \cos(\chi) T_{k+1}^{(\ell+1)}(\cos\chi)  - (\sin \chi)^{2}  T_{k+1}^{(\ell+2)}(\cos\chi)\right), 
\end{array}
\end{equation}
for all  $\alpha = (k, \ell, m) \in \Lambda$. \\
By using these expressions of  partial derivative of functions $(\YY_{\alpha})$  in  \eqref{formula_of_gradient_ww_bis}, 
we obtain a complete formula which is  readily available for practical use and for implementation.  
\subsection{Computational tests and numerical validation}\label{num_tests_sec} 
In this section, focus is on some numerical results that allow
to assess the practical  usability of the formula \eqref{u_expression_th} and \eqref{formula_nrg_th} and
the performances of the resulting numerical method
outlined in section \ref{mainRes_sect}. Three  different examples are investigated in the following. 
 In the first example we deal with non homogeneously magnetized spherical domain
 for which we have an error estimate by Theorem \ref{conver_theorem}. 
In the two last examples, the domain is homogeneously magnetized.  In all these three cases, 
 we derive expressions of the exact stray field,  to which the numerical 
solution is compared.  In all these computational  tests we set  $\mu_0=1$. 
\subsection*{Example 1: a non homogeneously magnetized sphere with $\Mg.\n = 0$} $\;$\\
We prefer starting numerical experiences with the case of a non homogeneously magnetized spherical sample, that is
 $$\omm =\{\x \in \R^3\bve  |\x| < \RAY\}$$  and 
 \begin{equation}\label{unitMg_ball}
\Mg = (\cos \theta) \e_\varphi + (\sin \theta) \e_\theta \mbox{ in } \omm. 
\end{equation}
It may be noted that $\Mg$ is complying with Heisenberg-Weiss constraint
\eqref{HeiWeiCtr} since $|\Mg| = 1$ in $\omm$. Besides,
$\Mg$ is tangential on the boundary of $\omm$ since
$\Mg . \n  = 0$ on $\pt \omm$ (here $\n[\x) = \x/|\x|$). \\
 We are able to give an analytical expression of the exact solution (see \cite{these_kaliche, KBoul_SF_IFEM}). More precisely, 
 \begin{equation}
\potf(\x) = \left\{
\begin{array}{ll}
\dps{ - \frac{2z}{9} +   \frac{2z}{3} \ln (\frac{|\x|}{ \RAY})}& \mbox{ if } |\x| \leq   \RAY, \;  \\
\dps{ -  \frac{2  \RAY^3 z}{9 |\x|^3} }  & \mbox{ if } |\x| \geq   \RAY.
\end{array}
\right.
\end{equation}
The exact stray-field energy is given by
\begin{equation}
\En_{sf}(\potf) =  \frac{\mu_0}{2} \int_{\R^3} |\grad \potf|^2 dx   = \frac{16}{81} \pi \RAY^3. 
\end{equation}
 Here we choose $\RAY = 1/2$. In Table \ref{table_res1}  we outline the computed stray-field energy \eqref{discrete_N_energy} for 
 several values of $N$ (considered as a discretization parameter).
 We also outline the relative $L^2$ error on the stray field $\Hd = - \grad \potf$ defined by 
$$
e_0(\Hd) = \frac{| \potf_{N}  - \potf|_{W^{1}_{0}(\Rt)}    }{ |\potf|_{W^{1}_{0}(\Rt)}}. 
$$
We can then observe that the error $e_0(\Hd)$  decreases in as $N^{-1.45}$. 
This is in accordance with Proposition \ref{third_main_th} in which it is forecasted that
 $$
 |\potf -\potf_{N}|_{W^{1}_{0}(\Rt)} \leq C N^{-1}  \|\div \Mg \|_{L^2(\omm)}. 
 $$
 Actually, the solution $u$ belongs $W^2_2(\R^3)$ since 
 $\div \Mg \in L^2(\omm)$ and $\Mg.\n=0$ on $\pt \omm$.  
  There is even a superconvergence with respect to this estimate. Note also that the error on the stray field energy decreases as $N^{-2.90}$ (in agreement with the identity  $| \EEST(\potf) - \EEST(\potf_N)| =  |\potf -\potf_{N}|^2_{W^{1}_{0}(\Rt)^{3}}$).

\begin{center}
\begin{table}
\begin{tabular}{|c|cccc|}
  \hline
N  &  $\EEST(u)$ & $\EEST(u_N)$ & $\dps{\frac{|\EEST(u)-\EEST(u_N)|}{\EEST(u)}}$ & $e_0(\Hd)$ \\
  \hline
10 & 0.07757018 & 0.07696625 & 7.78E-3 &  8.82E-2 \\
20 & -                  & 0.07750001  & 9.03E-4 &  3.00E-2 \\
30 & -                  & 0.07754315  & 3.48E-4 & 1.86E-2\\
40 & -                  & 0.07756016  &1.29E-4 &  1.13E-2\\
50 & -                 & 0.07756414  & 7.79E-5 &  8.82E-3  \\
60 &-                  &   0.07756708  &4.00E-5& 6.32E-3 \\
    \hline
 The log. slope& & & -2.90&-1.45  \\
  \hline
\end{tabular} 
  \caption{the exact and the approximate stray-field energy due to a non homogeneously magnetized sphere  (example 1).}\label{table_res1}
\end{table}
\end{center}
\subsection*{Example 2: a homogeneously magnetized sphere} $\;$\\
In this second benchmark test, we consider a spherical sample $\omm =\{\x \in \R^3\bve  |\x| < \RAY\}$ 
with a constant magnetization  $\Mg = \mz$.  It is easy to prove 
that the exact solution of \eqref{main_equa} is given by the formula:
\begin{equation}
\potf(\x) = \left\{
\begin{array}{ll}
\dps{ \frac{1}{3}{ \mz.\x } }& \mbox{ if } |\x| <  \RAY, \;  \\
\dps{\frac{ \RAY^3}{3}  \frac{\mz.\x}{|\x|^3} } & \mbox{ if } |\x| \geq   \RAY.
\end{array}
\right.
\end{equation}
The exact energy is 
\begin{equation}
\EEST(\potf) =  \frac{1}{2} \int_{\R^3} |\grad \potf|^2 dx =  - \frac{1}{2} \int_{\omm} \Mz.\vH dx  =  \frac{ 2 \pi |\mz|^2}{9}  \RAY^3.
\end{equation}
Here, we choose  $\mz = (0, 0, 1)$ and $\RAY = 0.5$. Thus, 
$$
\EEST(u) =   \frac{\pi}{36} = 0.08726646 
$$
It may be observed that  $[\frac{\pt \potf}{\pt n}] = -\mz.\n \ne 0$ on $\pt \omm$. Thus, $ \potf\not \in W^2_2(\Rt)$  although 
 $\potf_{|\omm} \in H^2(\omm)$ and $\potf_{|\Rt \backslash \bomm} \in W_2^2(\Rt \backslash \bomm)$
 (here $\potf_{|\omm}$ and $\potf_{|\Rt \backslash \bomm}$ designate the restrictions
 of $\potf$ to $\omm$ and to $\Rt \backslash \bomm$ respectively).  We are therefore 
 not within the validity assumptions of Theorem \ref{conver_theorem} and  the error estimates 
  \eqref{estima_norm_k1}  and  \eqref{estima_nrj_conv1} are no longer necessarily true. \\  
  In Table \ref{table_res2}, the approximate energy $\EESTH(\potf_N)$ is given for several values of the discretization parameter $N$. 
We also compute the relative $L^2$ error on the stray field $\Hd = - \grad \potf$.  One can observe that
this error decreases as $N^{-0.46}$. The error on the energy decreases as $N^{-0.93}$. \\
 Here again, convergence of the approximate solution to the exact one holds although the  
 normal component of $ \vH = -\grad \potf$  is not continuous  across  the boundary of the sample.
\begin{center}
\begin{table}
\begin{tabular}{|c|cccc|}
  \hline
N  &  $\EEST(u)$ & $\EEST(u_N)$ & $\dps{\frac{|\EEST(u)-\EEST(u_N)|}{\EEST(u)}}$ & $e_0(\Hd)$ \\
  \hline 
10 & 0.08726646 & 0.07845252&10.10E-2&0.3153 \\ 
20 &                     & 0.08252939&5.42E-2&0.2322  \\ 
30 &  -                 &  0.08402011&3.72E-2&0.1924   \\
40 & -                  &  0.08479348&2.83E-2&0.1680  \\
50 &  -                 &  0.08526692&2.29E-2&0.1511  \\
60 & -                  &  0.08558669 &1.92E-2 &0.1385  \\
    \hline
 {The log. slope}& & & -0.93&-0.46  \\
  \hline
\end{tabular} 
  \caption{The exact and the approximate stray-field energy due to an homogeneously magnetized  sphere (example 2).  }\label{table_res2}
\end{table}
\end{center}
\subsection*{Example 3:  homogeneously magnetized  cube. } $\;$\\
In this last test, 
we change the geometry of the sample and we consider a homogeneously magnetized  cubic rod  $\omm = ]-\gamma, \gamma[^3$, with $\gamma=1/2$,
and $\Mg = (0, 1, 0)$.  The stray-field energy in this case is (see, e. g., \cite{ClaasExl})
\begin{equation}
\EEST(\potf) = \frac{1}{6}. 
\end{equation}
The exact analytical expression of the demagnetizing field   is (see  \cite{Engel_Hesjedal}):
$$
\begin{array}{rcl}
\Hd(\x) &= &\dps{ \frac{1}{4\pi} \biggl( \sum_{k, \ell, m=1}^2 (-1)^{k+\ell+m} \ln(z+(-1)^m \gamma + \varrho \biggr) \e_x}, \\
 &-&  \dps{ \frac{1}{4 \pi} \biggl(  \sum_{k, \ell, m=1}^2(-1)^{k+\ell+m} \arctan\bigl(\frac{(x+(-1)^k \gamma) (z+(-1)^m \gamma)}{(y+(-1)^\ell \gamma) \varrho} \bigr) \biggr) \e_y}, \\
 &+& \dps{ \frac{1}{4 \pi}\biggl(   \sum_{k, \ell, m=1}^2 (-1)^{k+\ell+m} \ln(x+(-1)^k \gamma + \varrho) \biggr) \e_z}, 
\end{array}
$$
where $\varrho = \sqrt{(x+(-1)^k \gamma)^2 + (y+(-1)^\ell \gamma)^2 + (z+(-1)^m \gamma)^2}$.   \\
It may be observed that  $\Mg.\n \ne 0$ on $\pt \omm$. Thus, $ \potf \not \in W^2_2(\Rt)$ 
(see Remark \ref{regularity_rem}).
The numerical results summarized in Table  \ref{table_res3} confirm the convergence of the method and 
show that here too the $L^2$ error on the stray field $\Hd$ decreases as $N^{-0.47}$. The error on the energy decreases as $N^{-0.93}$, 
while the error on the energy decreases like  $N^{-0.97}$.

\begin{center}
\begin{table}
\begin{tabular}{|c|cccc|}
  \hline
N  &  $\EEST(u)$ & $\EEST(u_N)$ & $\dps{\frac{|\EEST(u)-\EEST(u_N)|}{\EEST(u)}}$ & $e_0(\Hd)$ \\
  \hline
 10 &  0.16666666       &  0.14711046  & 0.1173 &0.3397  \\
20 & -                           & 0.15617466   & 6.3E-2  & 0.2499  \\
30 & -                           & 0.15951131   & 4.2E-2. & 0.2066 \\
40 & -                           & 0.16123614   & 3.2E-2  & 0.180  \\
50 & -                           & 0.16229007   & 2.62E-2& 0.1618  \\
60 & -                           & 0.16300181   & 2.19E-2&0.1481  \\
    \hline 
 {The log. slope}&  & &-0.94 &-0.47  \\
  \hline
\end{tabular} 
  \caption{The exact and the approximate stray-field energy due to an homogeneously magnetized  cube  (example 3).  }\label{table_res3}
\end{table}
\end{center}

\section{Conclusion and perspectives}\label{conclusion_sec}
The formula  \eqref{u_expression_th}, in addition to being original, has several advantages both theoretically and numerically. 
From a computational  point of view, it has been established that the formula inspires a particularly efficient and easy to implement 
numerical method to calculate the demagnetizing field and the associated energy. Indeed, the numerical results show a rapid convergence of the method especially when $\Mg.\n = 0$ on $\pt \omm$.  In the latter case, the observed convergence is even faster than that predicted by the error estimate in Theorem \ref{conver_theorem} since the convergence in energy is of order close to $O(\frac{1}{N^3})$.  This suggests that these estimates are not optimal and could possibly be improved theoretically. In the case $\Mg.\n \ne 0$, the method also converges in accordance with the lemma, but one notes that  convergence of the  energy is of order close to $O(\frac{1}{N})$. This fact remains to be proven theoretically.  \\
$\;$\\
From a theoretical point of view, one could exploit  formula  \eqref{u_expression_th} to give a new expression to the functional to be minimized. 
Actually,  the total free energy can be expressed as: 
\begin{equation}\label{totalEnrg_bis}
\begin{array}{rcl}
E_{tot}(\Mg) &= & \dps{\alpha \int_{\omm} | \grad \Mg|^2 dx + \int_{\omm} \phi(\Mg) dx  - \mu_0 \int_{\omm}  \H_{ex} . \Mg dx} \\
&&+ \dps \sum_{k=0}^\infty \frac{2 \mu_0}{4(k+1)^2-1} \sum_{\alpha \in \INDS_k}  \left(\int_{\omm}\Mg.\grad \ww_{\alpha} dx\right)^2  + E_s. 
\end{array}
\end{equation}
It is well known that  the minimization of the functional $E_{tot}$ with respect
to the variable $\Mg$ under  Heisenberg-Weiss constraint  \eqref{HeiWeiCtr}  leads to the following partial differential
equation (see, e. g., \cite{hubert} and references therein):   
\begin{equation}\label{EL_equa_Magn}
\dps{ -2 \alpha \Delta \Mg + \grad_{\Mg} \phi (\Mg)- \mu_0 (\Hd+\H_{ext}) }= \dps{ \lambda \Mg}  \mbox{ in } \omm,  
\end{equation}
where $\lambda$ is a lagrangian multiplier. By sake of simplificity we assumed here that $E_s = 0$ (the reader can refer to, e. g.,
\cite{hubert} for the general equations taking into account this term). \\
Formula  \eqref{u_expression_th}  simplifies the system \eqref{EL_equa_Magn} 
and reduces it to only one equation
\begin{equation}
\begin{array}{l}
\hspace{1cm} -2 \alpha \Delta \Mg + \grad_{\Mg} \phi (\Mg)-  \mu_0 \H_{ext}  \\
\hspace{3cm}+ \; \dps \mu_0  \sum_{k=0}^\infty    \sum_{\alpha \in \INDS_k} \frac{2}{4(k+1)^2-1}  \left(\int_{\omm}\Mg.\grad \ww_{\alpha} dx\right)  \grad \ww_{\alpha} = \lambda \Mg \mbox{ in } \omm. 
 \end{array}
\end{equation}
The study of this non-local  PDE could provide new information about the best configuration 
minimizing the functional  $E_{tot}$.  If we truncate the serie on the left-hand side, keeping only the first term, we obtain the simplified approximate non local equation
\begin{equation}
 -2 \alpha \Delta \Mg + \grad_{\Mg} \phi (\Mg)-  \mu_0 \H_{ext} +\frac{2 \mu_0}{3 \pi^2 (|\x|^2+1)^{3/2}}  \biggl(\int_{\omm}\frac{\Mg.\x}{(|\x|^2+1)^{3/2}} dx\biggr)     \x = \lambda \Mg \mbox{ in } \omm. 
\end{equation}
The study of this kind of equations is beyond the scope of this paper; it will be the subject of a forthcoming paper.  \\
%
%
%
%
%

\appendix

\subsection*{A. Proof of Proposition \ref{formula_of_gradient_lem_bis}}
The objective here is to prove formula \eqref{formula_of_gradient_ww_bis}. Let $Y$ be an arbitrary smooth function defined
 on $\S^3$  and set 
$$
W(\x) =\left( \frac{2}{|\x|^2+1} \right)^{1/2} Y (\pi^{-1}(\x)), 
$$
(thus, if $Y = \YY_\alpha$, $\alpha \in \INDS$, then $W = \ww_\alpha$). 
In \cite{ararboulmezaoud1} and \cite{arar_boulmez_kk} (formula A.9), the authors prove the following identity (linking the gradient of 
$W$ to $Y$ and its tangential derivatives on the unit sphere): 
\begin{equation}\label{gradWgradY}
\grad W(\x)  = ( 1-\xi_4)^{1/2} \left(S(\xi)  \grad_{\xi} Y(\xi) - \frac{1}{2} Y(\xi) \hxi \right), \;\mbox{ for } \x \in \R^3,  
\end{equation} 
where  $\xi = \pi^{-1}(\x) \in  \S^3,$ $\hxi = (\xi_1, \xi_2, \xi_3)$ is the orthogonal projection of $\xi =  (\xi_1, \xi_2, \xi_3, \xi_4)$ on $\R^3$, 
$ \grad_{\xi} Y$ is the tangential gradient of $Y$ on $\S^3$  and $S(\xi)$ is the $3 \times 4$ rectangular
matrix 
\begin{equation}
\begin{array}{rcl}
S(\xi) &=&
\left(
\begin{array}{cccc}
1-\xi_4 & 0 & 0 & \xi_1 \\ 
0 & 1-\xi_4  & 0 &  \xi_2 \\ 
0 & 0 & 1-\xi_4  &  \xi_3 
\end{array}
\right) \\
& =& \left(
\begin{array}{cccc}
1-\cos \chi  & 0 & 0 & \dps{  \cos \phi \sin \theta \sin \chi } \\ 
0 & 1-\cos \chi  & 0 & \dps{ \sin \phi \sin \theta \sin \chi   } \\ 
0 & 0 & 1-\cos \chi   & \dps{\cos \theta  \sin \chi    }
\end{array}
\right). 
\end{array}
\end{equation}
It remains to spell out the expression of the tangential gradient $ \grad_{\xi} Y(\xi) $ in terms of 
partial derivatives of $Y$ with respect to $\phi$, $\theta$ and $\chi$, the spherical coordinates of $\xi$ (see section \ref{ABfunc_sec}). We state this
\begin{lemma}
If $\sin \chi \ne 0$, then 
\begin{equation}\label{gradFphithetachi}
\grad_{\!\xi} Y(\xi)  =  \frac{1}{\sin \chi} \!\! 
\left(\!\! \!
\begin{array}{cccc}
 -\dps{ \sin \phi  }    & \dps{ \cos \phi \cos \theta    } &  \cos \phi \sin \theta    \cos \chi \\
  \dps{ \cos \phi  }  & \dps{ \sin \phi \cos \theta  } &   \sin \phi \sin \theta  \cos \chi\\
 0 &-\dps{ \sin \theta  } &  \cos \theta  \cos \chi\\
0 &0   &  -\sin \chi
\end{array}
 \!\! \! \right)
\left( \!\! \!
\begin{array}{c}
 \dps{ \frac{1}{\sin \theta}  \frac{\pt  \tY}{ \pt  \phi} (\phi, \theta, \chi)  }  \\
 \dps{ \frac{\pt  \tY}{ \pt  \theta}(\phi, \theta, \chi) }\\
 \dps{ \sin \chi   \frac{\pt  \tY}{ \pt  \chi}( \phi, \theta, \chi) } 
\end{array}
 \!\! \! \right).
\end{equation}
where $\tY(\phi, \theta, \chi) = Y(\cos \phi \sin \theta \sin \chi, \sin \phi \sin \theta \sin \chi,  \cos \theta \sin \chi, \cos \chi)$. 
\end{lemma}
\begin{proof}
Consider the $0$-homogeneous function $F$ defined over $\R^4 \backslash \{0\}$ by
$$
F(\y) = Y(\frac{\y}{|\y|}), \; \y \in \R^4  \backslash \{0\}. 
$$
It follows that 
\begin{equation}
\grad_\xi Y(\xi) = \grad F(\xi) \mbox{ for  } \xi \in \S^3. 
\end{equation}
In view of Euler's homogeneous function lemma, we have
\begin{equation}\label{eulerThHomo}
\sum_{i=1}^4 y_i   \frac{\pt  F}{ \pt  y_i} (\y) = 0 \mbox{ for } \y \in \R^4 \backslash \{0\}. 
\end{equation}
Since 
$$
\tY(\phi, \theta, \chi) =  F(\cos \phi \sin \theta \sin \chi, \sin \phi \sin \theta \sin \chi,  \cos \theta \sin \chi, \cos\chi),
$$
we deduce that   
$$
\begin{array}{rcl}
\dps{  \frac{\pt  \tY}{ \pt  \phi}(\phi, \theta, \chi)   } &=& \dps{   \sin \theta \sin \chi \left(\!\!  -    \sin \phi   \frac{\pt  F}{ \pt  y_1}(\y)    +  \cos \phi   \frac{\pt  F}{ \pt  y_2}(\y)\!\! \right),  }  \\
\dps{  \frac{\pt  \tY}{ \pt  \theta} (\phi, \theta, \chi)  } &=& \dps{    \cos \theta \sin \chi     \left(\!\!   \cos \phi \frac{\pt  F}{ \pt  y_1}(\y)    +   \sin \phi    \frac{\pt  F}{ \pt  y_2}(\y) \!\! \right) } \\
&& \dps{  -\sin \theta \sin \chi     \frac{\pt  F}{ \pt  y_3} (\y),   }  \\
\dps{  \frac{\pt  \tY}{ \pt  \chi}(\phi, \theta, \chi)   } &=& \dps{   \sin \theta \cos \chi  \left(\!\!  \cos \phi  \frac{\pt  F}{ \pt  y_1} (\y)   +   \sin \phi    \frac{\pt  F}{ \pt  y_2}(\y)  \!\! \right) } \\
&&  +  \dps{   \cos \theta \cos \chi   \frac{\pt  F}{ \pt  y_3}(\y)  - \sin \chi    \frac{\pt  F}{ \pt  y_4}(\y),   }  \\
\end{array}
$$
where $\y = (\cos \phi \sin \theta \sin \chi, \sin \phi \sin \theta \sin \chi,  \cos \theta \sin \chi, \cos\chi)$. \\
Completing these identities with equation \eqref{eulerThHomo} gives a square linear system in
terms of the derivatives $ \frac{\pt  F}{ \pt  y_i}(\y) $, $1 \leq i \leq 4$.  Inverting this system
gives 
$$
\grad F(\y) = R(\xi) 
\left(
\begin{array}{c}
 \dps{   \frac{\pt  \tY}{ \pt  \phi} }  \\
 \dps{    \frac{\pt  \tY}{ \pt  \theta} }\\
 \dps{    \frac{\pt  \tY}{ \pt  \chi} } \\
   0,
\end{array}
\right)
$$
where
$$
R(\xi)  = 
\left( \!\!
\begin{array}{cccc}
 -\dps{ \frac{\sin \phi}{ \sin \theta \sin \chi  }   }    & \dps{ \frac{\cos \phi \cos \theta }{ \sin \chi }   } &  \cos \phi \sin \theta    \cos \chi& \cos \phi \sin \theta  \sin \chi  \\
  \dps{ \frac{ \cos \phi }{ \sin \theta  \sin \chi  } }  & \dps{ \frac{\sin \phi \cos \theta  }{ \sin \chi }   } &   \sin \phi \sin \theta  \cos \chi&\sin \phi  \sin \theta  \sin \chi \\
 0 &-\dps{ \frac{\sin \theta }{ \sin \chi }   } &  \cos \theta  \cos \chi&  \cos \theta  \sin \chi \\
0 &0   &  -\sin \chi & \cos \chi
\end{array}
 \!\! \right) 
$$
This ends the proof of \eqref{gradFphithetachi}. Formula \eqref{formula_of_gradient_ww_bis} is a direct consequence of \eqref{gradWgradY} and \eqref{gradFphithetachi}.
\end{proof}

\subsection*{B. The first few three-dimensional Arar-Boulmezaoud functions} $\;$\\
In this appendix we an give explicit formulas of the first few three-dimensional 
Arar-Boulmezaoud functions defined by \eqref{express_eigen1}.   These functions are illustrated
 in Table \ref{table_expli_func}.
\begin{table}\label{table_expli_func}
\begin{center}
\begin{tabular}{|c|c|c|c|}
  \hline
  $k$ & $\ell$ & $m$ & $\ww_{(k, \ell, m)}(\x)$   \\
 \hline 
0& 0& 0  & $\dps{  \frac{1}{ \pi (|\x|^2+1)^{1/2}  }}$\\
 \hline 
1& 0& 0  & $\dps{  \frac{2}{\pi}  \frac{|x|^2-1}{(|\x|^2+1)^{3/2}} }$   \\
 \hline
& 1& 0  & $\dps{  \frac{4}{\pi}  \frac{x_3}{(|\x|^2+1)^{3/2} }}$     \\
 \hline
& & 1  &   $\dps{  \frac{4}{\pi}  \frac{x_1}{(|\x|^2+1)^{3/2} }}$ \\
 \hline
& & -1  &   $\dps{  \frac{4}{\pi}  \frac{x_2}{(|\x|^2+1)^{3/2}} }$     \\ 
 \hline
2& 0& 0  &   $\dps{ \frac{1}{\pi}   } \frac{3|\x|^4-10 |\x|^2 +3}{(|\x|^2+1)^{5/2}}  $     \\
 \hline
& 1& 0  &  $\dps{  \frac{4 \sqrt{6}}{\pi}  \frac{x_3(|\x|^2-1)}{(|\x|^2+1)^{5/2}} }$      \\
 \hline
& & 1 &   $\dps{  \frac{4\sqrt{6}}{\pi}  \frac{x_1(|\x|^2-1)}{(|\x|^2+1)^{5/2}} }$       \\
 \hline
& & -1  &   $\dps{  \frac{4\sqrt{6}}{\pi}  \frac{x_2(|\x|^2-1)}{(|\x|^2+1)^{5/2}} }  $     \\
 \hline
& 2& 0  & $\dps{  \frac{4\sqrt{2}}{\pi} \frac{3 x_3^2 -  |\x|^2 }{(|\x|^2+1)^{5/2}}  }$  \\
 \hline
& & 1 &  $\dps{  \frac{8\sqrt{6}}{\pi} \frac{x_1 x_3 }{(|\x|^2+1)^{5/2}}  }$    \\
 \hline
& & 2  & $\dps{  \frac{4\sqrt{6}}{\pi} \frac{x^2_1 - x_2^2 }{(|\x|^2+1)^{5/2}}  }$     \\
 \hline
& & -1 &   $\dps{  \frac{8\sqrt{6}}{\pi} \frac{x_2 x_3 }{(|\x|^2+1)^{5/2}}  }$    \\
 \hline
& & -2  &   $\dps{  \frac{8\sqrt{6}}{\pi} \frac{x_1x_2 }{(|\x|^2+1)^{5/2}}  }$    \\
  \hline
\end{tabular}
\end{center}
\caption{Explicit expressions of the first Arar-Boulmezaoud functions in $\R^3$}
\end{table}

$\;$\\
{\bf Declarations}. \\
{\it Conflict of interest}: The author declares no competing interests.

\bibliographystyle{plain}
\bibliography{micromagAB.bib}

\end{document}